%% file: root.tex
\documentclass[letterpaper, 10 pt, conference]{ieeeconf}  

\newif\iftechrep
\global\techrepfalse
\global\techreptrue

\IEEEoverridecommandlockouts                              

\input{packages}

\title{\LARGE \bf
Coalition Formation with Limited Information Sharing for Local Energy Management*
}

\author{Luke Rickard$^{1}$, Paola Falugi$^{1}$ and Eric C. Kerrigan$^{2}$
\thanks{*This work was supported by EPSRC grant EP/Z536106/1}
\thanks{$^{1}$University of East London $^{2}$Imperial College London}%
}

\begin{document}

\maketitle
\thispagestyle{empty}
\pagestyle{empty}

\begin{abstract}

Distributed energy systems with prosumers require new methods for coordinating energy exchange among agents. 
Coalitional control provides a framework in which agents form groups to cooperatively reduce costs; however, existing bottom-up coalition-formation methods typically require full information sharing, raising privacy concerns and imposing significant computational overhead.

In this work, we propose a limited information coalition-formation algorithm that requires only limited aggregate information exchange among agents. 
By constructing an upper bound on the value of candidate coalitions, we eliminate the need to solve optimisation problems for each potential merge, significantly reducing computational complexity while limiting information exchange. 
We prove that the proposed method guarantees cost no greater than that of decentralised operation.

Coalition strategies are optimised using a distributed approach based on the Alternating Direction Method of Multipliers (ADMM), further limiting information sharing within coalitions. 
We embed the framework within a model predictive control scheme and evaluate it on real-world data, demonstrating improved economic performance over decentralised control with substantially lower computational cost than full-information approaches.
\end{abstract}

\input{sections/1-introduction}
\input{sections/2-energy_distribution}
\input{sections/3-coalitions}
\input{sections/4-MPC}
\input{sections/5-results}
\input{sections/6-conclusion}
\balance

%








\bibliographystyle{ieeetr}
\bibliography{bibliography}

\end{document}

%% file: packages.tex
\let\labelindent\relax
\usepackage{enumitem}
\usepackage{hyperref}
\usepackage{balance} 

\usepackage{graphics} 
\usepackage{epsfig} 
\usepackage{mathptmx} 
\usepackage{times} 
\usepackage{amsmath} 
\usepackage{amssymb}  
\usepackage{amsfonts}
\usepackage{xcolor}

\usepackage[noend]{algpseudocode}
\usepackage{algorithm}

\newtheorem{definition}{Definition}
\newtheorem{assumption}{Assumption}
\newtheorem{corollary}{Corollary}
\newtheorem{theorem}{Theorem}

\usepackage{subcaption}
\usepackage{tikz}

\DeclareMathOperator*{\argmin}{arg\,min}
\newcommand{\G}{\mathsf{G}}
\newcommand{\allu}{\mathsf{u}}
\newcommand{\C}{\mathsf{C}}

%% file: sections/1-introduction.tex
\section{Introduction}

Energy production around the world is currently undergoing a significant paradigm shift from fossil fuels to renewable energy sources.
One part of this transition is the change from relying on a few large, high-capacity power plants to depending on many smaller, widely distributed generation units~\cite{NADEEM2023101096}. 
In particular, we are interested in the rise of ``prosumers'', participants in the energy market who both consume and produce electricity~\cite{FILHO2024100158,Parag2016}.
This motivates exploring novel control techniques better suited to the new energy generation landscape.

One promising area of research that leverages this shift in generation is \emph{coalitional control}~\cite{7823093,DBLP:conf/amcc/FeleMC15}.
Coalitional control is an approach that considers grouping entities involved in a large problem into separate coalitions that can be solved separately. 
We refer to this initial grouping of agents, before the subproblems are solved, as coalition formation.
For the grid scenario, coalitional control can involve buildings agreeing to trade energy with one another, aiming to pay less than they would by interacting only with the grid.


Forming coalitions allows agents to reduce their individual costs without incurring large computational or communication overheads, and without sharing full local information with agents outside the coalition (once coalitions are formed). 
Revealing private information (such as usage data) in the case of the smart grid, can lead to vulnerabilities, allowing attackers to infer various household traits~\cite{giaconi2018privacy,razavi2018rethinking}.

However, while recent work that has explored bottom-up coalition formation~\cite{DBLP:conf/amcc/FeleMC15} addresses computational complexity, it still requires full information sharing during coalition formation, negating these privacy benefits. 
In particular, the technique in~\cite{DBLP:conf/amcc/FeleMC15} required solving the complete optimisation problem for every pair of agents, requiring each agent to broadcast all optimisation variables and constraints to all other agents. 
Not only does this introduce privacy concerns, but it also represents a computational bottleneck. 
We seek a ``limited information sharing'' scheme that avoids sharing full local demand, generation, storage, and constraint data. 
Instead, agents only share tentative grid trade plans during coalition formation, and coalition trade plans afterwards. 
This reduces information shared and computational costs.

We make use of Model Predictive Control (MPC) techniques to optimise energy trades and battery charging profiles, with forecasted energy usage and production for a day ahead. 
MPC has been studied in this context in~\cite{oldeWurtel2012mpc}, but it is limited to a single building, whereas we consider a group of buildings that may form into coalitions. 
Our work extends upon current extensions of MPC to coalitions~\cite{baldivieso2021coalitional,sanchez2023robust} by reducing the amount of information shared (and computational requirements) for each possible coalition formation step, and dynamically updates coalitions at every time step to account for changing conditions.

Coalitional control aligns with techniques in distributed control, such as the Alternating Direction Method of Multipliers~\cite{DBLP:journals/ftml/BoydPCPE11}, explored for energy grids in~\cite{doi:10.1049/PBPO149E_ch5,DBLP:journals/tsg/ZhangHDDX18}. 
While distributed control provides a powerful framework for optimising systems composed of multiple agents, it typically assumes a fixed interaction structure and does not address how agents should be grouped. 
In many applications, forming such groups is desirable to limit information sharing, reduce computational burden, or manage communication costs. 
Coalitional control addresses this by determining groups of agents (coalitions) that can then optimise their decisions jointly, e.g., using distributed optimisation techniques.

In this work, we propose a coalition formation framework that reduces both computational complexity and information sharing. The approach enables agents to evaluate potential coalition merges using only limited, aggregate information, without requiring full knowledge of local models or solving a full optimisation problem for each candidate merge. The main contributions of this paper are as follows:

\begin{itemize}

\item We show that exact optimisation of candidate coalition merges can be replaced by a valid upper bound on coalition value, avoiding the need to solve an optimisation problem for every candidate merge.
\item For local energy management, we construct an efficiently computable upper bound requiring only aggregate grid-trade information, thereby reducing information disclosure.
\item We prove that the resulting coalition-formation algorithm terminates in a finite number of iterations and produces a coalition structure with total cost no greater than that of decentralised operation.
\item We integrate the proposed coalition-formation scheme with distributed optimisation via ADMM and MPC.
\item Numerical results on real-world data demonstrate improved performance over decentralised control and lower computational cost than full-information methods.

\end{itemize}

While we focus on energy markets in this paper, our technique is generally applicable to all trading games, and we discuss how it may be extended to allow a trade-off between optimality and information sharing. 
All proofs are available in~\cite{TechRep}. 

\subsection{Notation}

Where we have a sequence $\{G(t)\}_{t=1}^T$, we use $\G$ to refer to the entire sequence collectively (i.e. $\G=\{G(t)\}_{t=1}^T$).

\begin{figure}
\centering
\includegraphics[width=\linewidth]{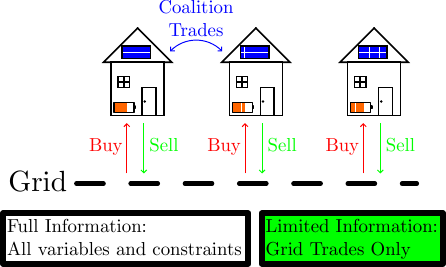}
\caption{Overview of the energy system with coalitions. 
Buildings may trade with the grid, or with other buildings within their coalition (but not those outside the coalition). 
Our limited-information approach shares only planned grid trades, whereas full-information methods require sharing all optimisation variables and constraints.}
\label{fig:schematic}
\end{figure}

%% file: sections/2-energy_distribution.tex
\section{Background}
\subsection{Building Model}

We consider a collection of $N$ buildings $\mathcal{M}_1,\dots,\mathcal{M}_N$, each with some fixed demand and generation profiles $d_i(t), g_i(t)\in\mathbb{R}_{\geq 0}$, corresponding to the energy consumed, and produced, at each time step $t$ across the time horizon $\mathcal{T}=\{1,\dots,T\}$.
Each building has a battery with initial state $\mathrm{SoC}_i(1)=\mathrm{SoC}_i^\mathrm{init}$, whose state $\mathrm{SoC}_i(t) \in [0,\overline{\mathrm{SoC}}_i]$ evolves according to the following dynamics
    \begin{equation}
    \label{eq:SoC}
        \mathrm{SoC}_i(t+1) = \mathrm{SoC}_i(t)+\rho^c_iu^c_i(t) - \frac{1}{\rho^d_i}u^d_i(t),
    \end{equation}
    where $\rho^c_i,\rho_i^d \in (0,1)$ are charge and discharge efficiencies respectively, and $u_i = (u^c_i,u^d_i) \in [0,\overline{u}_i]^2$ are charge and discharge signals. 
    The building may trade energy with the grid by purchasing or selling energy to the grid at each time step $G_i(t)= (G_i^\mathrm{buy}(t),G_i^\mathrm{sell}(t)) \in \mathbb{R}^2_{\geq 0}$, and we refer to the total trade with the grid as $G^\mathrm{tot}_i(t) = G^\mathrm{buy}_i(t)-G^\mathrm{sell}_i(t)$.
    At each time step, the building $i$ must satisfy a power balance
    \begin{equation}
    \label{eq:power_balance}
        d_i(t)+u^c_i(t)=g_i(t)+u^d_i(t)+G^\mathrm{tot}_i(t).
    \end{equation}

We then define a linear cost function for each building as
\begin{equation}
    L_i(\G_i) = \sum_{t=1}^T \left(P^\mathrm{buy}(t)G_i^\mathrm{buy}(t)-P^\mathrm{sell}(t)G_i^\mathrm{sell}(t)\right),
\end{equation}
where $P^\mathrm{buy}(t),P^\mathrm{sell}(t)$ are the prices offered by the grid for buying and selling electricity. 
We assume $P^\mathrm{buy}(t)>P^\mathrm{sell}(t)>0$ for all $t$.  
Under this assumption, and with $\rho_c,\rho_d \in (0,1)$, simultaneous charging and discharging is suboptimal~\cite{DBLP:journals/corr/abs-2104-06267}.
Finally, we denote the feasible set of variables $\{\G_i,\allu_i\}$ by 
\begin{align}
        \mathfrak{X}_i=\{&\G_i\in\mathbb{R}^{2\times T}_{\geq0},\allu_i\in[0,\overline{u}_i]^{2\times T} \colon\nonumber\\
        &d_i(t)+u^c_i(t)=g_i(t)+u^d_i(t)+G^\mathrm{tot}_i(t)\;\forall t\in \mathcal{T},\nonumber\\
            &\mathrm{SoC}_i(t+1) = \mathrm{SoC}_i(t)+\rho^c_iu^c_i(t) - \frac{1}{\rho^d_i}u^d_i(t) \; \forall t\in \mathcal{T},\nonumber\\
            &\mathrm{SoC}_i(t) \in [0,\overline{\mathrm{SoC}}_i]\; \forall t\in \mathcal{T}\cup\{T+1\},\\
            &\mathrm{SoC}_i(1)=\mathrm{SoC}_i^\mathrm{init}\}\nonumber.
\end{align}

\subsection{Open-loop Optimisation}

We first seek to minimise the cumulative cost for all buildings $\sum_{i=1}^N L_i(\G_i)$ in an open-loop fashion, we later consider how to extend to model predictive control in Section~\ref{sec:MPC}.
We therefore define an optimisation problem
\begin{equation}
    \begin{aligned}
        \min_{\{\G_i,\allu_i\}\in \mathfrak{X}_i, i=1,\ldots,N}&\sum_{i=1}^N L_i(\G_i)\\
    \end{aligned}
\end{equation}

\subsubsection{Decentralised}


We can naturally separate this problem into $N$ linear programs  
\begin{equation}
    \begin{aligned}
        \sum_{i=1}^N \min_{\{\G_i,\allu_i\}\in \mathfrak{X}_i}&L_i(\G_i)\\
    \end{aligned}
\end{equation}
which can then be solved directly using off-the-shelf optimisation packages. 
However, we can improve upon this optimum by allowing for buildings to form into groups, which we call coalitions, in which energy is freely traded to benefit all members of the coalition.

\subsection{Coalitions}

We consider the problem of forming a coalitional structure, defined in the sequel, to reduce the shared cost of all agents.

\begin{definition}[Coalition Structure]
\label{def:CS}
A coalition structure $\mathcal{C}=\{\mathcal{C}_1,\dots,\mathcal{C}_M\}$, denotes a disjoint partitioning of the complete set of building indices
\begin{equation}
    \bigcup_{j=1}^M\mathcal{C}_j=\{i\}_{i=1}^N, \qquad \mathcal{C}_i\cap\mathcal{C}_j=\emptyset, \ i,j = 1,\dots,M, i\neq j.
\end{equation}
\end{definition}
Note that a coalition $\mathcal{C}_j$ may be a collection of buildings, or a single building.

We now introduce the additional optimisation variables $C_i^\mathrm{buy}(t),C_i^\mathrm{sell}(t)\in \mathbb{R}_{\geq 0}$ for each building, which allow that building to buy from, and sell to, its coalition. 
Again, we use $C_i^\mathrm{tot}(t) = C_i^\mathrm{buy}(t)-C_i^\mathrm{sell}(t)$ to refer to the total traded with the coalition, and $C_i(t)=(C_i^\mathrm{buy}(t),C_i^\mathrm{sell}(t))$ to refer to the pair. 
The power balance in~\eqref{eq:power_balance} is then amended to  
    \begin{equation}
    \label{eq:power_balance_in_coal}
        d_i(t)+u^c_i(t)=g_i(t)+u^d_i(t)+G_i^\mathrm{tot}(t)+C_i^\mathrm{tot}(t),
    \end{equation}
    and the constraint set $\mathfrak{X}_i$ is amended similarly (with other terms unaffected).

For a given coalition $\mathcal{C}_j \in \mathcal{C}$, we can then define the optimisation problem as
\begin{equation}\label{eq:opt_in_coal_ind}
    \begin{aligned}
        V(\mathcal{C}_j)=\min_{\{\G_i,\allu_i,\C_i\}\in\mathfrak{X}_i, i\in\mathcal{C}_j}&\sum_{i\in\mathcal{C}_j} L_i(\G_i)\\
        \mathrm{subject \; to}& \;  \sum_{i\in\mathcal{C}_j} C_i^\mathrm{tot}(t)=0, \; \forall t\in\mathcal{T},
    \end{aligned}
\end{equation}
with a new constraint to maintain a power balance for the coalitional trades.
For the entire coalition structure, we then seek to minimise the cumulative value $\sum_{j=1}^MV(\mathcal{C}_j).$
    

We now have $M\leq N$ linear programs, with the new constraint preventing further separation.
We assume that trades within a coalition are settled at 0 cost; this allows us to focus on the limited information coalition formation mechanism and derive a simple achievable upper bound. 
Extensions to unequal internal trading costs (e.g. to account for transmission costs) are left for future work.




%% file: sections/3-coalitions.tex
\section{Coalition Formation}
\label{sec:coal_form}

In order to optimise the cumulative cost of all agents, it is also necessary to optimise over coalition structures. 
We now consider a number of techniques for performing this optimisation, with varying degrees of information sharing. 
We assume that there is an upper bound on the size of a coalition~$\mathcal{C}_\mathrm{max}$; this constraint may be introduced to limit the number of agents with whom operational information is shared, or to lower communication/computational overhead.

\subsection{Optimal (Size-Limited) Coalition}

First, we consider the problem of directly obtaining the true optimal coalition, hence we seek to solve 
\begin{equation}
\label{eq:coal_brute}
    \begin{aligned}
        \min_{\mathcal{C}}&\sum_{\mathcal{C}_j\in\mathcal{C}}V(\mathcal{C}_j)\\
        \mathrm{subject \; to} & \;\bigcup_{j=1}^M\mathcal{C}_j=\{i\}_{i=1}^N\\
        & \;\mathcal{C}_j\cap\mathcal{C}_k=\emptyset,\; j,k = 1,\dots,M, k\neq j\\
        &\;|\mathcal{C}_j| \leq \mathcal{C}_\mathrm{max}, \;j=1,\dots,M,
    \end{aligned}
\end{equation}
which can be formulated as a mixed-integer linear program (due to the optimisation over coalition structures $\mathcal{C}$), since we now optimise also over each building's membership of coalitions. 
However, this approach is combinatorial in the number of agents (requiring iteration over all possible partitions of a set).

\subsection{Bottom-Up Coalition Formation~\cite{DBLP:conf/amcc/FeleMC15}}

One less computationally expensive approach to approximately solve the coalition formation is proposed in~\cite{DBLP:conf/amcc/FeleMC15}. 
The algorithm is summarised in Algorithm~\ref{alg:bottom_up}.
\begin{algorithm}[tb!]
\caption{Bottom-Up Coalition Formation}\label{alg:bottom_up}
\begin{algorithmic}[1]
\State $\mathcal{A}_\mathrm{new} \gets \{\mathcal{M}_1,\dots,\mathcal{M}_N\}$
\Repeat
\State $\mathcal{A} \gets \mathcal{A}_\mathrm{new}$
\ForAll{$a \in \mathcal{A}$}
\State $\{\G^\ast_i, \allu^\ast_i, \C^\ast_i\}_{i\in a} \gets \argmin_{\{\G_i,\allu_i,\C_i\}_{i\in a}}\sum_{i \in a}L_i(\G_i)$\Comment{Subject to $\{\G_i,\allu_i,\C_i\} \in \mathfrak{X}_i$ and $\sum_{i\in a}C_i^\mathrm{tot}(t)=0,t\in\mathcal{T}$}
\EndFor
\ForAll{$\tilde{\mathcal{C}} \in \{(\mathcal{A}_j\cup\mathcal{A}_k)\vert j = 1,\dots,|\mathcal{A}|-1,k=j+1,\dots,|\mathcal{A}| \} $}
\If{$|\tilde{\mathcal{C}}|\leq$ max coalition size}
\State \label{ln:opt} $\mathrm{val}(\tilde{\mathcal{C}}) \gets V(\mathcal{A}_j)+V(\mathcal{A}_k)-V(\tilde{\mathcal{C}})$
\EndIf
\EndFor
\State $\mathrm{sortedCosts}\gets\mathrm{sort}(\mathrm{val})$
\State $\mathcal{A}_\mathrm{new} \gets \emptyset$
\For{$(\tilde{\mathcal{C}},\mathrm{val})$ in $\mathrm{sortedCosts}$}
    \If{val$(\tilde{\mathcal{C}})>0$ and $\tilde{\mathcal{C}} \cap (\bigcup \mathcal{A}_\mathrm{new}) = \emptyset$}\label{ln:val_check} \Comment{Coalition beneficial and neither agent accounted for}
    \State $\mathcal{A}_\mathrm{new} \gets \mathcal{A}_\mathrm{new} \cup \{\tilde{\mathcal{C}}\}$
        \EndIf

\EndFor
\State $\mathcal{A}_\mathrm{new} \gets \mathcal{A}_\mathrm{new} \cup (\mathcal{A}\setminus(\bigcup \mathcal{A}_\mathrm{new}))$ \label{ln:add_unpaired}\Comment{Add all unpaired agents/coalitions}
\Until{$\mathcal{A}_\mathrm{new}=\mathcal{A}$}
\end{algorithmic}
\end{algorithm}

This algorithm cannot guarantee convergence to the optimal coalition structure, but is significantly less computationally intensive than solving \eqref{eq:coal_brute} directly. 
The number of possible coalition structures of $N$ agents grows according to the Bell numbers and is therefore combinatorial in $N$. 
Consequently, obtaining the optimal coalition structure by exhaustive enumeration quickly becomes intractable. 
Algorithm~\ref{alg:bottom_up} avoids this combinatorial search by iteratively evaluating candidate pairwise merges among current coalitions. 
If $m_k$ denotes the number of coalitions at iteration $k$, then the number of candidate merges at that iteration is $\binom{m_k}{2}=O(m_k^2)$. 


\subsection{Limited Information Coalition Formation}

The approach in Algorithm~\ref{alg:bottom_up} suffers from a crucial drawback, namely the need to solve the complete optimisation problem for every pair of agents, thereby sharing full local information with everyone. 
We wish to avoid sharing data on individual agents' consumption and generation behaviour. 
Not only does this raise privacy concerns, but also introduces a computational burden for every potential coalition.
To remove this requirement, we recognise that \emph{only an upper bound} on the solution to \eqref{eq:opt_in_coal_ind} is needed.
We denote by $\overline{V}(\mathcal{C}_j)$, the upper bound on the joint cost received by coalition $\mathcal{C}_j$
\begin{equation}
\label{eq:priv_UB}
    \begin{aligned}
                \overline{V}(\mathcal{C}_j) \geq V(\mathcal{C}_j).
    \end{aligned}
\end{equation}
Then, we replace Line~\ref{ln:opt}, with
\begin{equation}
    \mathrm{val}(\tilde{\mathcal{C}}) \gets V(\mathcal{A}_j)+V(\mathcal{A}_k)-\overline{V}(\tilde{\mathcal{C}}),
\end{equation}
and proceed with the rest of the algorithm unchanged, forming coalitions only if this new bound is positive.

The tighter this bound (by sharing more information), the better the final coalition structure. 
Hence, this trade-off gives agents a choice between information sharing and optimality. 
Once coalitions are formed, we assume the following (in Section~\ref{sec:ADMM}, we discuss how this can be achieved without sharing private information):
\begin{assumption}[Value Knowledge of Formed Coalitions]
    Once a coalition $\mathcal{C}_j$ is formed, its true value $V(\mathcal{C}_j)$ can be evaluated exactly.
\end{assumption}
We then have:
\begin{theorem}[Limited Information Algorithm Properties]
    \label{thm:improve}
    Consider Algorithm~\ref{alg:bottom_up}, with an upper bound $\overline{V}(\mathcal{C}_j)$ (satisfying \eqref{eq:priv_UB}) on the joint cost in place of solving the optimisation problem in Line~\ref{ln:opt}.
    Then the algorithm will satisfy the following properties
    \begin{enumerate}[wide, labelwidth=!, labelindent=0pt]
        \item The resulting coalition structure satisfies Definition~\ref{def:CS};
        \item The algorithm terminates in a finite number of iterations;
        \item The resulting coalition structure has a value at most equal to the decentralised optimum.
    \end{enumerate}
\end{theorem}
\iftechrep
\begin{proof}
We consider each item in turn:
    \begin{enumerate}[wide, labelwidth=!, labelindent=0pt]
    \item Since a coalition is accepted only if it does not overlap with a coalition already in the coalition structure (Line~\ref{ln:val_check}), and any coalitions not in a pair with another coalition are added at the end (Line~\ref{ln:add_unpaired}), then the resulting coalition structure is valid.
    \item We initialise the algorithm with $M=N$ coalitions (each agent in a singleton coalition). 
    At each iteration, we either form new coalitions (strictly decreasing $M$) or terminate. 
    $M$ is bounded below by $\lceil\frac{N}{\mathcal{C}_\mathrm{max}}\rceil$ and decreases if coalitions form, so the algorithm terminates in at most $N-\lceil\frac{N}{\mathcal{C}_\mathrm{max}}\rceil$ iterations.
        \item Our proposed algorithm forms coalitions between pairs of coalitions $a,b$ only if the upper bound satisfies
    \begin{equation}
    \label{eq:improvement}
    \overline{V}(a \cup b)<V(a)+V(b).
    \end{equation} 
    By assumption, $\overline{V}(a \cup b)$ is a valid upper bound which satisfies \eqref{eq:priv_UB}.
    Combining these inequalities, we have 
    \begin{equation}
        V(a \cup b) \leq \overline{V}(a \cup b) < V(a) + V(b),
    \end{equation}
    so that any formed coalition offers a strict improvement in total cost compared to agents working individually. 
    Since our proposed algorithm forms coalitions iteratively, starting from the decentralised solution, and each coalition formation strictly reduces the total cost, the final coalition structure will have a total cost at most equal to the decentralised solution. 
    Equality holds only when no coalitions are formed.
    \end{enumerate}
    \end{proof}

\fi

\subsubsection{A Practical Bound}
A very computationally cheap option for this bound, that is also non-trivial, is to share only our current planned net grid trades (calculated in a decentralised manner). We then compare the grid consumption at each time step to check if some gain can be made.
First, define
\begin{equation}
    Q_j(t)=\sum_{i\in\mathcal{A}_j}G^\mathrm{tot}_i(t),
\end{equation}
the total grid trade by a coalition.
Then, check if the sequence
\begin{equation}
    \Xi_{jk}(t)=
  Q_j(t) 
  Q_k(t) ,
\end{equation}
contains any negative values (corresponding to one coalition buying (positive $Q$) and the other selling (negative $Q$) at the same time step).

Then, the upper bound on coalition cost $V(\mathcal{A}_j\cup\mathcal{A}_k)$ is 
    \begin{align}
           & V(\mathcal{A}_j)+V(\mathcal{A}_k)    \label{eq:private_val}
\\&
           -\sum_{t\colon\Xi_{jk}(t)<0}\left(P^\mathrm{buy}(t)-P^\mathrm{sell}(t)\right)\min\left\{\left|  Q_j(t)\right|,\left| Q_k(t)\right|\right\}, \nonumber
    \end{align}
i.e. the sum across time steps of possible trades, if no further optimisation is performed. 
From this we can see that
\begin{equation}
    \mathrm{val}(\tilde{\mathcal{C}})=\sum_{t\colon\Xi_{jk}(t)<0}\left(P^\mathrm{buy}(t)-P^\mathrm{sell}(t)\right)\min\left\{\left| Q_j(t)\right|,\left| Q_k(t)\right|\right\}
\end{equation}

\begin{corollary}
    The bound provided in \eqref{eq:private_val} is a non-trivial upper bound to the coalitional cost. 
\end{corollary}
\iftechrep
\begin{proof}
    We show that the quantity in \eqref{eq:private_val} is achievable by explicitly constructing a feasible coalition strategy whose cost is exactly no greater than that bound. 
    Since the optimal coalitional cost is no larger than the cost of any feasible strategy, this proves that \eqref{eq:private_val} is an upper bound.
    We first prove the result for two individual buildings $i$ and $j$. The extension to general coalitions follows by linearity of the aggregate net grid exchange, with variables replaced by their coalitional aggregate values (e.g. $Q_i(t)$ in place of $G_i(t)$).
    
    We have that the strategies $\G_i,\allu_i$ and $\G_j,\allu_j$ must already be feasible strategies for agent $i$ and $j$, since they are solutions to the decentralised optimisation (with $\C_i=\C_j=0$).
    We fix $\allu_i$ and $\allu_j$ and do not alter these, and instead focus on shifting trades from the grid to the coalition.
    
    For each time step $t$, if $\Xi_{ij}(t)<0$, one building is buying from the grid while the other is selling at the same time step.
    In the decentralised solution, this results in simultaneous grid purchase at price  $P^\mathrm{buy}(t)$ and grid sale at price $P^\mathrm{sell}(t)$. 
    Thus, we consider moving $\Delta(t)=\min\{|G^\mathrm{tot}_i(t)|,|G^\mathrm{tot}_j(t)|\}$ trade from the grid to the coalition.
    By instead directly exchanging $\Delta(t)$ units of energy between the two buildings, we eliminate both the corresponding grid purchase and sale, yielding a cost reduction of $\Delta(t)(P^\mathrm{buy}(t)-P^\mathrm{sell}(t))$.
    Assume, without loss of generality, that $G^\mathrm{tot}_i(t)>0$, we then have $C^\mathrm{tot}_i(t) = \Delta(t)$, and $C^\mathrm{tot}_j(t) = -C^\mathrm{tot}_i(t)$ so that $C^\mathrm{tot}_i(t)+C^\mathrm{tot}_j(t)=0$, satisfying our coalition power balance. 
    Since $(G_i(t),u_i(t))$ was feasible, and $\hat{G}^\mathrm{tot}_i(t)+C^\mathrm{tot}_i(t)=G^\mathrm{tot}_i(t)$, then the modified strategy $(\hat{\G}_i,\allu_i,\C_i)$ also satisfies the power balance (with demand and generation profiles unchanged).
    
    We now only modify our trades as $\hat{G}^\mathrm{tot}_i(t)=G^\mathrm{tot}_i(t)-C^\mathrm{tot}_i(t)$, and $\hat{G}^\mathrm{tot}_j(t)=G^\mathrm{tot}_j(t)+C^\mathrm{tot}_i(t) = G^\mathrm{tot}_j(t)-C^\mathrm{tot}_j(t)$. 
    Thus, at time step $t$, the combined cost reduction is
    \begin{align}
        &P^\mathrm{buy}(t)\left[G^\mathrm{tot}_i(t)-\hat{G}^\mathrm{tot}_i(t)\right]+P^\mathrm{sell}(t)\left[G^\mathrm{tot}_j(t)-\hat{G}^\mathrm{tot}_j(t)\right]\nonumber\\
        &=\Delta(t)(P^\mathrm{buy}(t)-P^\mathrm{sell}(t)),
    \end{align}
    where the quantity on the left is exactly equal to the cost reduction for buildings $i$ and $j$. 
    After the modification, building $i$ reduces its grid purchase by $\Delta(t)$, and building $j$ reduces its grid sale by $\Delta(t)$.
    By summing across time steps where $\Xi_{ij}(t)<0$, we obtain feasible coalitional cost 
    $$V(\{i\}) + V(\{j\})    - \sum_{t : \Xi_{ij}(t) < 0}
    \left(P^{\mathrm{buy}}(t) - P^{\mathrm{sell}}(t)\right)
   \Delta(t).$$
    Since the optimal coalitional cost $V({i,j})$ is no larger than the cost of any feasible coalition strategy, we obtain the bound in \eqref{eq:private_val}.
    Thus, this bound is achievable and is a valid upper bound on the coalitional cost.
\end{proof}
\fi
We assume zero internal trading costs, in the future we plan to extend to this case by introducing a penalty term in the upper bound which accounts for this.

\subsection{Within Coalition Optimisation}
\label{sec:ADMM}

Once a coalition structure is formed, we must also decide how to optimise the joint strategy of the coalition. 
Solving the optimisation problem centrally raises similar information sharing concerns to Algorithm~\ref{alg:bottom_up}, as it requires full information from all agents.
To address this, we adopt a distributed solution based on ADMM~\cite{DBLP:journals/ftml/BoydPCPE11}, which enables agents to compute their local decisions while coordinating only through limited information exchange. 
We reformulate the coalition optimisation problem in a consensus form, allowing each agent to solve a local subproblem while enforcing agreement on shared variables through a coordinator. 
Full details are available in~\cite{TechRep}.

\iftechrep
We introduce auxiliary variables
$\mathbf{z}_i  \in \mathbb{R}^T$ for each agent $i$, representing the
consensus coalition trade schedule across the prediction horizon, and dual variables $\mathbf{\lambda}_i \in \mathbb{R}^T$.
These variables enforce the coalition power balance constraint $\sum_i z_i(t)=0$ at each time step while allowing agents to solve their optimisation problems locally.
For each agent $i$ in coalition $\mathcal{C}_l$, we make use of the following update steps
\begin{subequations}
    \begin{align}
    \G_{i,k+1}&,\allu_{i,k+1},\C_{i,k+1}\label{eq:ADMM_par}\\
    &= \argmin_{\{\G_i,\allu_i,\C_i\}\in\mathfrak{X}_i }L_i(\G_i)+{\mathbf{\lambda}_{i,k}}^\top \mathbf{C}^\mathrm{tot}_i+\frac{c}{2}\|\mathbf{C}^\mathrm{tot}_i-\mathbf{z}_{i,k}\|^2\nonumber\\
    \mathbf{z}_{i,k+1} &= \mathbf{C}^\mathrm{tot}_{i,k+1} - \frac{1}{|\mathcal{C}_l|}\sum_{j\in\mathcal{C}_l} \mathbf{C}^\mathrm{tot}_{j,k+1}\label{eq:ADMM_cent}\\
    \mathbf{\lambda}_{i,k+1} &= \mathbf{\lambda}_{i,k}+c(\mathbf{C}^\mathrm{tot}_{i,k+1}-\mathbf{z}_{i,k+1})\label{eq:ADMM_dual},    
\end{align}
\end{subequations}
where \eqref{eq:ADMM_par} is performed in parallel by all agents, and (\ref{eq:ADMM_cent},~\ref{eq:ADMM_dual}) are performed centrally (the coordinator initialises $\mathbf{\lambda}_i$ and shares the updated value at each iteration). 
We use $\mathbf{C}^\mathrm{tot}_i$ to denote the stacked vector of coalitional trades for agent $i$, and $c$ is a user-selected regularisation parameter (guidance on selecting this parameter is available in~\cite{DBLP:journals/ftml/BoydPCPE11}).
Note that this requires sharing only tentative coalition trades at each iteration, thereby ensuring that individual consumption and generation data is not disclosed. 
Since the problem is convex and can be written in a consensus form, standard ADMM convergence results apply \cite{DBLP:journals/ftml/BoydPCPE11}

Since ADMM enforces the coalition balance constraint only asymptotically, early termination may result in violations of the constraint $\sum_{i\in \mathcal{C}_j} \mathbf{C}_i^{\mathrm{tot}} = \mathbf{0}$. To ensure feasibility, we take the coalition trades from $\mathbf{z}_{i}$, which satisfy this constraint by construction. 
The local grid trades $G_i$ are then recomputed to restore the power balance constraint for each agent, while satisfying their non negativity constraints and keeping the control inputs $u_i$ fixed. This yields a feasible solution, at the expense of suboptimality.
The local power balance is restored by updating the grid trades while keeping the control inputs $\allu_i$ unchanged, ensuring that the power-balance constraints are satisfied at each time step.
\fi

%% file: sections/4-MPC.tex
\section{Model Predictive Control}
\label{sec:MPC}

Our methods up to this point have dealt with open-loop optimisation of the energy purchasing problem. 
We now consider applying our methods in a model predictive control (MPC) framework. 

We introduce predicted consumption and generation profiles $\hat{d}(t+k|t), \hat{g}(t+k|t)$ for $k = 0,\dots,\min\{T_\mathrm{pred},T-t\}$, where $(\cdot|t)$ denotes a value predicted at step $t$. 
We assume that energy prices and battery dynamics are known precisely, so that this prediction is the only form of uncertainty. 

Then, for every time step, we solve the open-loop optimisation problem using this predicted behaviour. 
We utilise the first planned battery charge signal $u(t)$ and (if applicable) the first coalition trade $C(t)$, and then purchase (or sell) energy as necessary to satisfy the power balance in \eqref{eq:power_balance}. 
Note that, if our predictions are inaccurate, we may trade a different amount of energy with the grid than planned. 
Once we employ a strategy at time step $t$, we uncover the realised stage cost
\begin{equation}
    R(\mathcal{C}_j)(t) = \sum_{i\in\mathcal{C}_j}P^\mathrm{buy}(t)G^\mathrm{buy}_i(t)-P^\mathrm{sell}(t)G^\mathrm{sell}_i(t).
\end{equation}

When employing a coalition formation algorithm, we use predicted data to solve the coalition formation problem. 
Coalitions are reformed at every time step to adapt to changing conditions, and, after coalitions are formed, we check if any violate the following inequality
\begin{equation}
\label{eq:MPC_cond}
    R(\mathcal{C}_j)(t) \leq \sum_{i \in \mathcal{C}_j} R(\{i\})(t).
\end{equation}
If this condition is violated, then this coalition is dissolved and constituent agents work individually. 
In this way, we guarantee that our algorithm performs no worse than the decentralised optimum.
As in the decentralised case, if predictions are inaccurate, we trade different amounts with the grid, maintaining coalition trades and charging signal.

\begin{theorem}[MPC with Coalition Formation]
Consider the predictions $\hat{d}(t|t)=d(t)$, $\hat{g}(t|t)=g(t)$. Let
 $\{\mathcal{C}_j(t)\}_{j=1}^{M(t)}$ be the coalition structure formed at time $t$, and define
\begin{align}
J_{\mathrm{coal}}^{\mathrm{MPC}} &= \sum_{t=1}^T\sum_{j=1}^{M(t)} R(\mathcal{C}_j(t))(t)\\
J_{\mathrm{dec}}^{\mathrm{MPC}}&= \sum_{t=1}^T\sum_{i=1}^N R(\{i\})(t).
\end{align} 
Where, in the coalitional case, any coalitions not satisfying \eqref{eq:MPC_cond} are dissolved and agents work individually on that time step.
Then, despite coalitions being reformed every time step,
\begin{equation}
    J_{\mathrm{coal}}^{\mathrm{MPC}} \leq J_{\mathrm{dec}}^{\mathrm{MPC}}.
\end{equation}
\end{theorem}
\iftechrep
\begin{proof}
At each time step $t$, we enforce: 
\begin{equation}
    R(\mathcal{C}_j(t))(t) \leq \sum_{i \in \mathcal{C}_j(t)} R(\{i\})(t).
\end{equation}
Summing over all time steps and coalitions we have
\begin{equation}
    \sum_{t=1}^T\sum_{j=1}^{M(t)} R(\mathcal{C}_j(t))(t) \leq  \sum_{t=1}^T\sum_{j=1}^{M(t)}\sum_{i \in \mathcal{C}_j(t)} R(\{i\})(t),
\end{equation}
which gives
   $ J_{\mathrm{coal}}^{\mathrm{MPC}} \leq J_{\mathrm{dec}}^{\mathrm{MPC}}$.



\end{proof}
\fi


%% file: sections/5-results.tex
\section{Numerical Results}
\label{sec:exp}

We make use of a real-world dataset~\cite{EMSx,EMSx_article} to perform experiments, with energy usage recorded for 15-minute time slots, predictions for both generation and demand, battery data and energy prices taken from this dataset\footnote{Code is available at \url{https://github.com/lukearcus/Coalition_grid}.}.
Whenever we seek to solve \eqref{eq:opt_in_coal_ind} for $V(\mathcal{C}_j)$, we perform ADMM until the primal residual is less than $10^{-5}$, and use a penalty parameter $c=\frac{1}{2}$, with SCS as numerical solver. 
The experiments thus focus on assessing the impact of limited-information coalition formation, rather than differences in the underlying optimisation solver.
We use a coalition size limit that allows computations to be completed within a reasonable time frame, while offering reasonable cost reductions (for our limited information approach, we compare different coalition sizes in Section~\ref{sec:sizes}).
All simulations were performed using 2 CPU cores at 2.6 GHz and 4 GB RAM.

\subsection{Open Loop Optimisation}
Results are presented in Table~\ref{tab:results} for the open loop optimisation problem using first all 70 buildings, and then 8 buildings, with a maximum coalition size of 7 buildings, and planning for a single day ahead. 
For 70 agents, the optimal coalition problem was terminated after being unable to find a solution within 24 hours. 

As can be seen, our algorithm offers an improvement over the decentralised scheme in both cases, and is even able to perform optimally when solving with 8 buildings. 
Despite the Bottom-Up algorithm outperforming our algorithm for the complete set of buildings, we will see in Section~\ref{sec:scalability}, this comes at a cost of increased computation, and greater information sharing. 
The brute force approach to identify the true optimum coalition requires substantial computation (as discussed in Section~\ref{sec:coal_form}), and was unable to finish within 24 hours for 70 buildings.

\begin{table}[tb]
    \centering
    \caption{Mean cost per building (averaged over the N buildings) for open-loop (TO denotes timeout)}
    \label{tab:results}
    \begin{tabular}{c|c|c}
        & \multicolumn{2}{c}{Cost (euros)}\\
        Scheme & 70 Buildings & 8 Buildings \\ \hline
        Decentralised & 239 & 454 \\
        Centralised & 101 & 430 \\
        Limited Information & 232 & 430 \\
        Bottom-Up  & 157 & 431 \\
        Optimal Coalition & TO & 430 \\
    \end{tabular}
\end{table}

\subsection{Model Predictive Control}

We now compare our techniques against existing methods when implementing a model predictive controller. 
In Table~\ref{tab:results_MPC} we compare the cost achieved, number of ADMM iterations per time step (including both for coalition formation, and within-coalition solving) and runtime, for 10 buildings across the full horizon (1 day) with a 2-hour prediction horizon, and with coalitions restricted to a maximum of 6 buildings. 
We see that our limited information algorithm achieves a lower cost than the decentralised scheme, without incurring a significant computational burden. 
In Figure~\ref{fig:use_compar}, we see how the formation of coalitions affects trades with the grid. 
Under the coalition formation scheme, trades follow a similar envelope but spike toward zero when trades are beneficial as agents trade with each other, not the grid. 
Once again, the brute force approach to find the optimal coalition was too computationally intensive to terminate in time.

\begin{figure*}[h!]
    \centering
    \begin{subfigure}[t]{.5\linewidth}
        \centering
        \includegraphics[width=\linewidth]{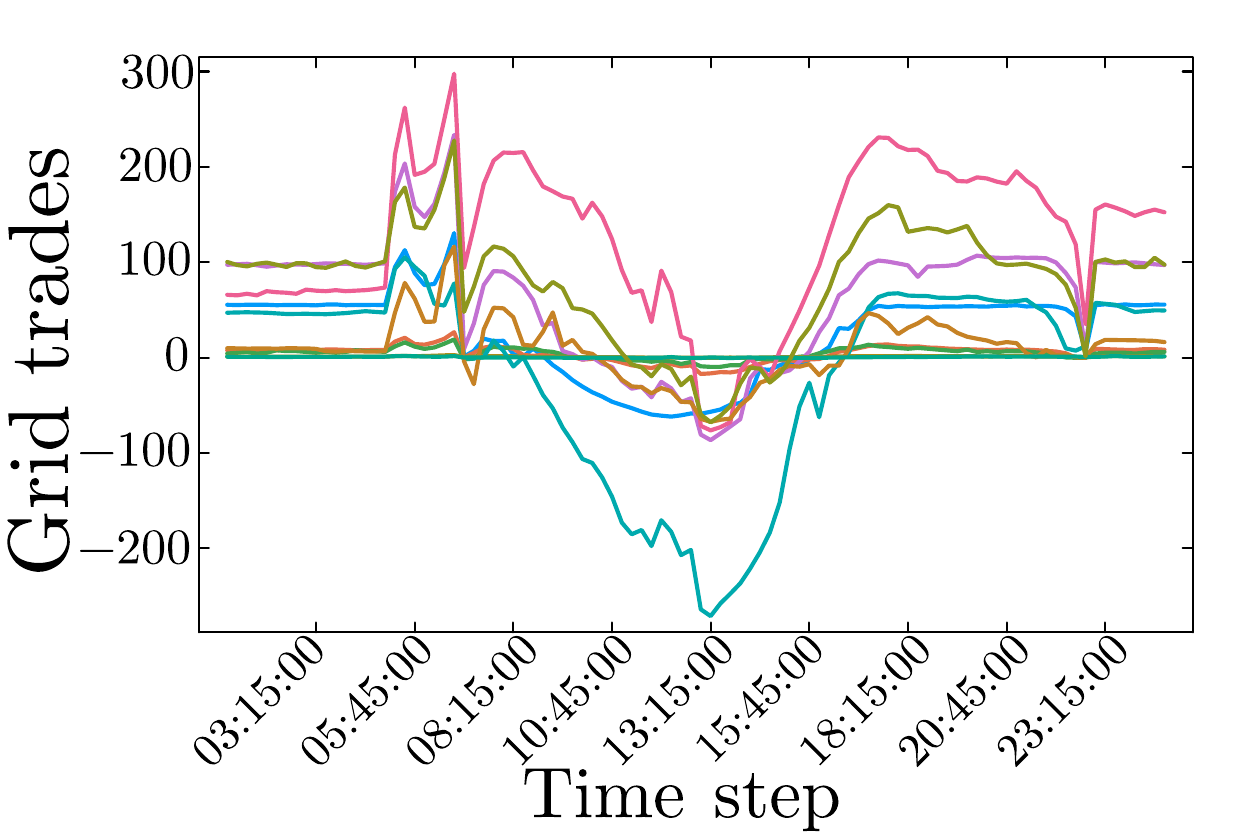}
        \caption{Decentralised}
    \end{subfigure}%
        \begin{subfigure}[t]{.5\linewidth}
        \centering
        \includegraphics[width=\linewidth]{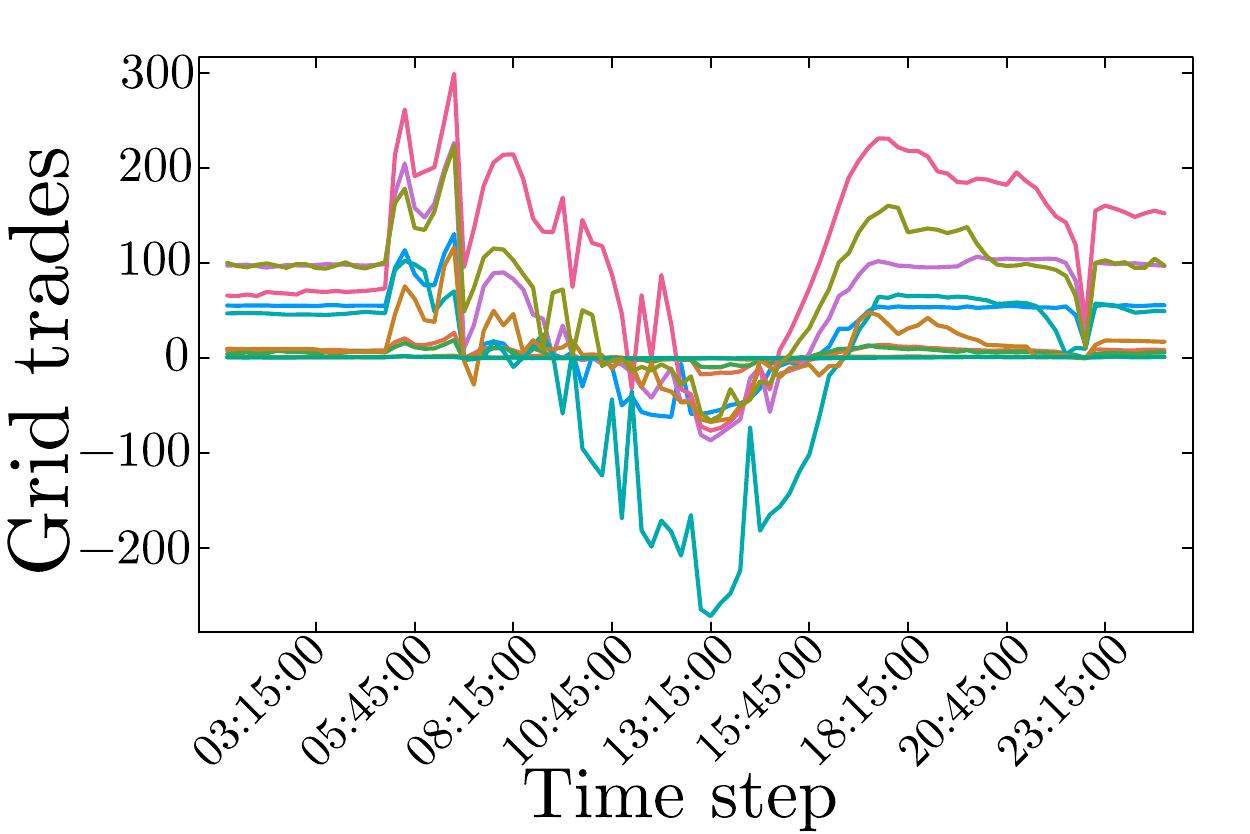}
        \caption{Limited information coalitions}
    \end{subfigure}
    \caption{Grid trades (in kWh, positive buy, negative sell) for a day under decentralised and limited information coalitions.}
    \label{fig:use_compar}
\end{figure*}

\begin{table}[tb]
    \centering
    \caption{Mean cost per building (averaged over the 10 buildings) and mean ADMM iterations (averaged over timesteps) for MPC (TO denotes timeout)}
    \label{tab:results_MPC}
    \begin{tabular}{c|c|c|c}
        Scheme & Cost (euros) & ADMM Iterations & Runtime (s) \\ \hline
        Decentralised & 496 & - & 17 \\
        Centralised & 475 & 14986 & 50700 \\
        Limited Information & 489 & 1362 & 4067 \\
        Bottom-Up & 478 & 18652 & 63432 \\
        Optimal Coalition & TO & - & - \\
    \end{tabular}
\end{table}

\subsection{Scalability}
\label{sec:scalability}

In Figure~\ref{fig:var_builds}, we provide a comparison of the average number of ADMM iterations needed per time step for a varying number of buildings $N$, using the same setup for MPC as the previous section. 
We allowed a maximum of 72 hours total computation time for each algorithm. 
As can be seen, our limited information coalition formation technique outperforms both other algorithms as the number of buildings becomes larger, and is still able to maintain relatively few iterations until convergence, even with 24 buildings. 

\begin{figure}[h]
    \centering
    \includegraphics[width=\linewidth]{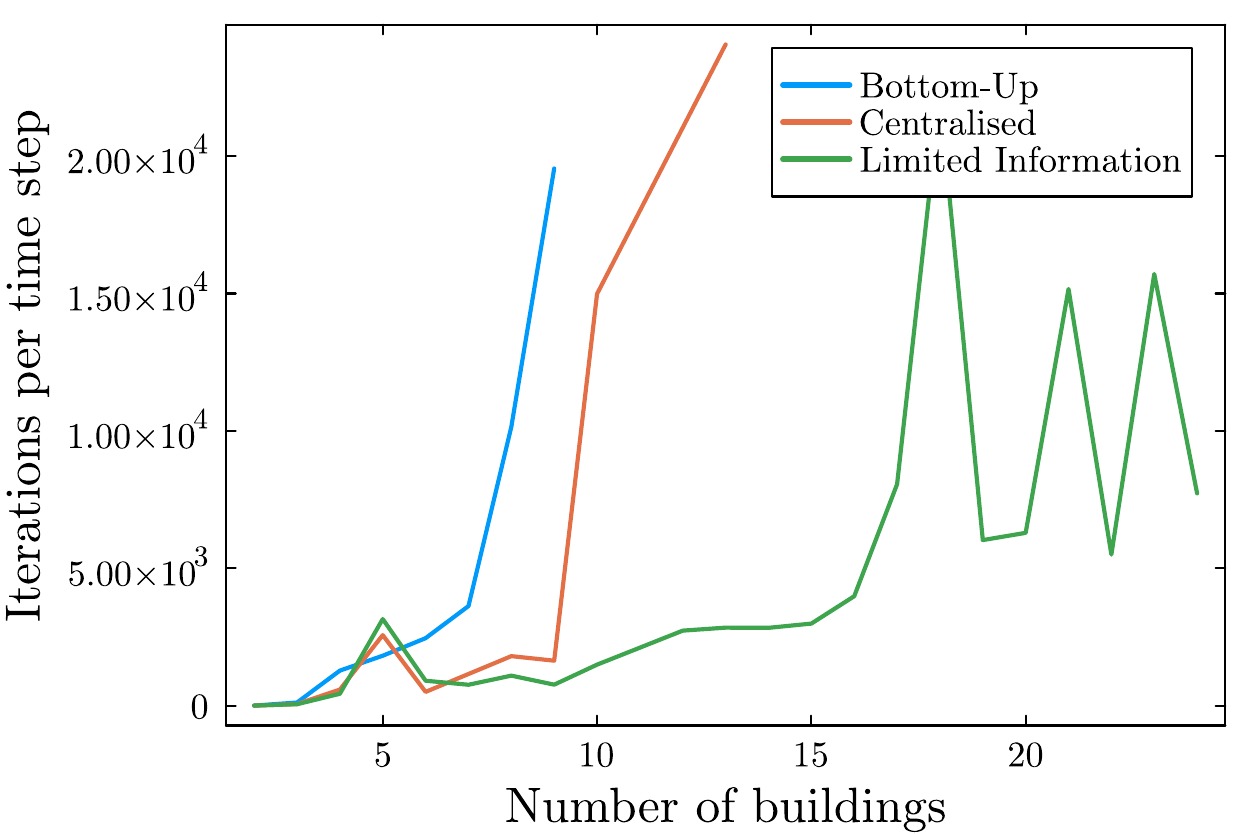}
    \caption{Number of ADMM iterations vs Number of buildings for different schemes}
    \label{fig:var_builds}
\end{figure}

\subsection{Coalition Size}
\label{sec:sizes}

For our proposed limited information scheme, we compare the effects of different choices for maximum coalition size in Figure~\ref{fig:size}. 
We used 20 buildings, a time horizon of 1 day, and a 2-hour lookahead.
As we expect, the larger the coalition the more computationally intensive the problem becomes, but the lower the average cost. 
This suggests that a trade-off between coalition size and cost should be sought. 
For our previous experiments, we identified a maximum coalition size of six as offering a reasonable trade-off.
\begin{figure}[hb]
    \centering
    \includegraphics[width=\linewidth]{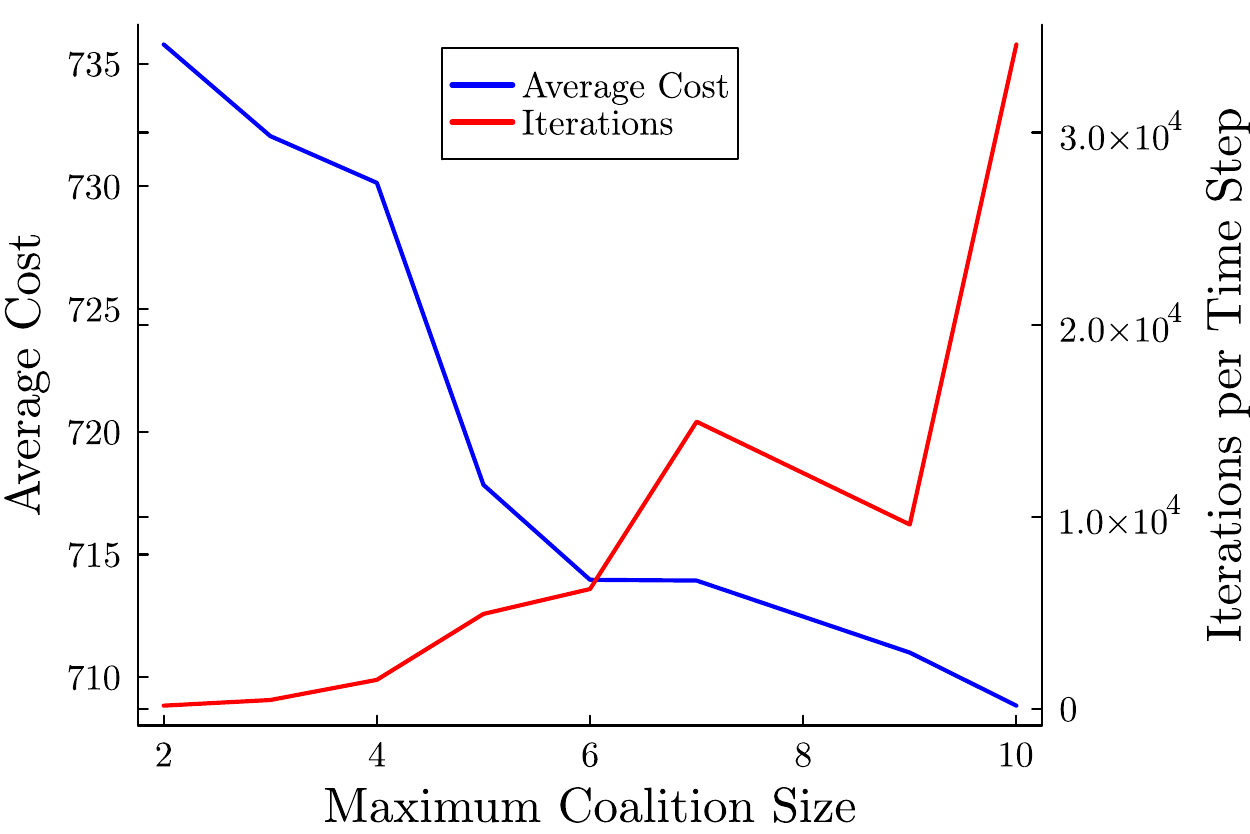}
    \caption{Effect of maximum coalition size on number of ADMM iterations, and average cost.}
    \label{fig:size}
\end{figure}

%% file: sections/6-conclusion.tex
\section{Conclusions}

Our results demonstrate that coalitional formation is achievable with limited information sharing while maintaining improved performance over the decentralised case. 
These preliminary results raise several avenues for future work, including analysing the trade-off between information sharing and optimality, exploring stronger guarantees than improvement over the decentralised optimum, extending to more general coalition formation problems, addressing uncertainty in future predictions for MPC, and providing formal privacy guarantees and analysis.